\documentclass[draftclsnofoot,12pt,onecolumn]{IEEEtran}
\usepackage{bbm}
\usepackage{amssymb,amsmath,amsthm}
\usepackage{graphicx,epstopdf,psfrag,url,cite,xcolor,multirow,float}
\usepackage{caption-patch}
\usepackage[small,bf]{caption}
\usepackage[font=footnotesize]{subfig}
\usepackage{bm}

\usepackage{algorithm}
\usepackage{algorithmic}
\usepackage[super]{nth}

\newtheoremstyle{mythm}
{\topsep}   
{\topsep}   
{\itshape}      
{0pt}       
{\bfseries} 
{:}         
{5pt plus 1pt minus 1pt}    
{\thmname{#1}\thmnumber{ #2}\thmnote{ (#3)}}
\theoremstyle{mythm}

\newtheorem{proposition}{Proposition}

\newtheorem{observation}{Observation}

\newcommand{\singlesize}{0.6}

\IEEEoverridecommandlockouts

\begin{document}

\title{Robust Beamforming for IRS-assisted Wireless Communications under Channel Uncertainty}
\author{
Yongchang Deng, Yuze Zou, Shimin Gong, Bin Lyu, Dinh Thai Hoang, and Dusit Niyato
\thanks{Yongchang Deng and Shimin Gong are with the School of Intelligent Systems Engineering, Sun Yat-sen University, China. Yuze Zou is with the School of Electronic Information and Communications, Huazhong University of Science and Technology, China. Bin Lyu is with the Key Laboratory of Ministry of Education in Broadband Wireless Communication and Sensor Network Technology, Nanjing University of Posts and Telecommunications, China. Dinh Thai Hoang is with the Faculty of Engineering and Information Technology, University of Technology Sydney, Australia. Dusit Niyato is with School of Computer Science and Engineering, Nanyang Technological University, Singapore. \emph{(Corresponding author: Shimin Gong)}}
}

\maketitle
\thispagestyle{empty}

\begin{abstract}
In this paper, we consider IRS-assisted transmissions from a multi-antenna access point (AP) to a receiver with uncertain channel information. By adjusting the magnitude of reflecting coefficients, the IRS can sustain its operations by harvesting energy from the AP's signal beamforming. Considering channel estimation errors, we model both the AP-IRS channel and the AP-IRS-receiver as a cascaded channel by norm-based uncertainty sets. This allows us to formulate a robust optimization problem to minimize the AP's transmit power, subject to the user's worst-case data rate requirement and the IRS's worst-case power budget constraint. Instead of using the alternating optimization method, we firstly propose a heuristic scheme to decompose the IRS's {phase shift} optimization and the AP's active beamforming. Based on semidefinite relaxations of the worst-case constraints, we further devise an iterative algorithm to optimize the AP's transmit beamforming and the magnitude of the IRS's reflecting coefficients efficiently by solving a set of semidefinite programs. Simulation results reveal that the AP requires a higher transmit power to deal with the channel uncertainty. Moreover, the negative effect of channel uncertainty can be alleviated by using a larger-size IRS.
\end{abstract}

\section{Introduction}

Recently, intelligent reflecting surface (IRS) has been proposed to enhance wireless communications by proactively reconfiguring the radio environment, i.e., wireless channels~\cite{18scm_ian1,19survey_renzo}. This provides an additional degree of freedom for system design and performance optimization of wireless networks. As such, the integration of transmission control at end devices and the online reconfiguration of wireless channels can be envisioned as a revolutionary technology for future wireless networks. The operation of IRS relies on a large array of passive scattering elements with the physical dimension equivalent to the signals' wavelength. Each scattering element can induce a controllable phase shift on the incident RF signals. All passive elements are connected and jointly controlled by the IRS controller, which can communicate with external user devices to exchange control information, e.g.,~\cite{joint_overview,18scm_ian3}.

To enhance wireless communications, the IRS controller can optimize the phase shifts of all scattering elements jointly and thus reshape the physical wireless channels between the transceivers. In contrast to the conventional transmit beamforming, this is viewed as passive beamforming due to the passive operations of the IRS's scattering elements. Recently, there is an upsurge in research works on the optimization of IRS-assisted wireless networks by leveraging the IRS's passive beamforming capability, e.g.,~channel capacity or energy efficiency maximization in~\cite{powermax-swipt,Huang2018Large}, transmit power minimization in~\cite{Wu2018Intelligent,18pbf_rui2}, and physical layer security issues in~\cite{Yu2019Enabling,wang2019eefficient}. However, most of the existing works assume perfect channel state information (CSI) and focus on one-shot or static beamforming optimization problems. This requires the capabilities of channel sensing and signal processing involving the IRS, which becomes very challenging due to the IRS's passive nature.
As a countermeasure, the authors in~\cite{ss_asym} modeled the channel's error estimates due to pilot contamination as Gaussian variables, and then asymptotically studied the negative effect of channel estimation errors in IRS-assisted uplink transmissions.

Instead of such a stochastic approach, another approach for modeling channel uncertainty is to impose a bounded norm on the error estimates, which generally results in robust optimization formulations to ensure the worst-case performance guarantee. The authors in~\cite{frameworkconf} considered a norm-based uncertainty set for the reflecting channels from the IRS to receivers and proposed a robust power minimization problem subject to the worst-case data rate requirements at individual users. A similar uncertainty model is applied to an IRS-assisted secure communication system in~\cite{robustsecure}, where the channel between the IRS and an eavesdropper is subject to a norm-based uncertainty set. The authors in~\cite{robustsecure} proposed a robust sum-rate maximization problem subject to the worst-case information leakage to an eavesdropper. Different from~\cite{frameworkconf}, an extended work in~\cite{robustframework} assumed that the cascaded channel from the base station to the receiver via the IRS is subject to a norm-based uncertainty set. The worst-case robust design in~\cite{robustframework} ensures that each user can have a data rate guarantee for all realizations of the channel conditions. The solutions to the above robust designs typically rely on convex reformulations of the worst-case constraints in the first place, and then use the alternating optimization (AO) method to optimize iteratively the active and passive beamforming strategies.

In this paper, we consider a similar norm-based uncertainty model for IRS-assisted multiple-input single-output (MISO) transmissions from a multi-antenna access point (AP) to a receiver. Different from~\cite{ss_asym,frameworkconf,robustsecure,robustframework}, we focus on a more practical case with both the worst-case data rate requirement at the receiver and the worst-case power budget constraint at the IRS, considering the uncertainties in both the AP-IRS and the cascaded AP-receiver channels. We propose a robust design to minimize the AP's transmit power by jointly optimizing the AP's transmit beamforming and the IRS's passive beamforming strategies. To the best of our knowledge, this is the first work to study the worst-case power budget constraint in IRS-assisted wireless networks. To solve this robust problem, we firstly propose a heuristic scheme to decompose the active and passive beamforming optimizations. We then focus on the joint optimization of the AP's transmit beamforming and the magnitude of the IRS's reflecting coefficients. By exploiting the problem structure, we design an efficient search algorithm to solve the joint optimization problem in an iterative manner. Our simulation results reveal that, with uncertain channel information, a significant increase in the AP's transmit power is required to maintain the same quality of service provisioning to the user. Another countermeasure to the negative effect of channel uncertainty is to use a larger-size IRS instead of increasing AP's transmit power. This verifies the potential benefit for improving energy efficiency by deploying IRS in wireless networks.

\section{System model}\label{sec:model}

We consider an IRS-assisted MISO downlink system as shown in Fig.~\ref{fig_system_model}. The IRS has $N$ reflecting elements and the multi-antenna AP has $M$ antennas serving one single-antenna receiver. The system model can be easily extended to a multi-user case. We assume that the number of passive reflecting elements is much larger than the number of the AP's antennas. The IRS controller is capable of adjusting the phase shift and the magnitude of each reflecting element dynamically according to the channel conditions. The joint control of phase shifts and magnitudes, namely passive beamforming, provides the capability of reshaping the physical channel conditions in favor of information transmissions from the AP to the receiver. The direct AP-receiver, AP-IRS and IRS-receiver complex channels are denoted by ${\bf g}\in \mathbb{C}^{M\times 1}$, ${\bf H}\in \mathbb{C}^{M\times N}$ and ${\bf f}\in \mathbb{C}^{N\times 1}$, respectively.

\subsection{IRS-assisted Channel Enhancement}
We assume that the IRS controller can adjust the magnitude and phase of each reflecting coefficient individually. Each reflecting element sets a phase shift $\theta_n\in[0,2\pi]$ and its magnitude $\rho_n\in[0,1]$ to reflect the incident RF signals. Let $\bm{\Theta}=\text{diag}(\rho_1 e^{j\theta_1},\ldots,\rho_N e^{j\theta_N})$ denote the IRS's passive beamforming, where $\text{diag}({\bf a})$ denoting the diagonal matrix with the diagonal vector ${\bf a}$. Hence the IRS-assisted equivalent channel from the AP to the receiver is given by
\begin{equation}\label{eq_enhanced}
\hat{\bf g}={\bf g}+{\bf H}{\bf \Theta}{\bf f},
\end{equation}
where ${\bf H}= [{\bf h}_1,\ldots,{\bf h}_N]$ denotes the channel matrix from the AP to the IRS. Let ${\bf w} \in\mathbb{C}^{M \times 1}$ denote the AP's beamforming vector and $s$ denote the complex symbol transmitted by the AP with unit power. The received signal at the receiver is given by $y=\hat{\bf g}^H{\bf w}s+\nu_d$~\footnote{Here the superscript $H$ denotes the conjugate transpose.}, where $\nu_d \sim \mathcal{CN}(0, \sigma^2)$ is the Gaussian noise with zero mean and variance $\sigma^2$. For notational convenience, we normalize the noise variance to unit one. Hence, the received signal-to-noise ratio (SNR) is given by
\begin{equation}\label{eq_SNR}
\gamma({\bf w}, {\boldsymbol{\Theta}})=\lVert ({\bf g}+{\bf H}{\bf \Theta}{\bf f})^H{\bf w}\rVert^2 .
\end{equation}
It is clear that the SNR performance depends on the AP's transmit beamforming ${\bf w}$ and the IRS's passive beamforming ${\bf \Theta}$, which are coupled in a non-convex form.

\begin{figure}[t]
	\centering
	\includegraphics[width=\singlesize\textwidth]{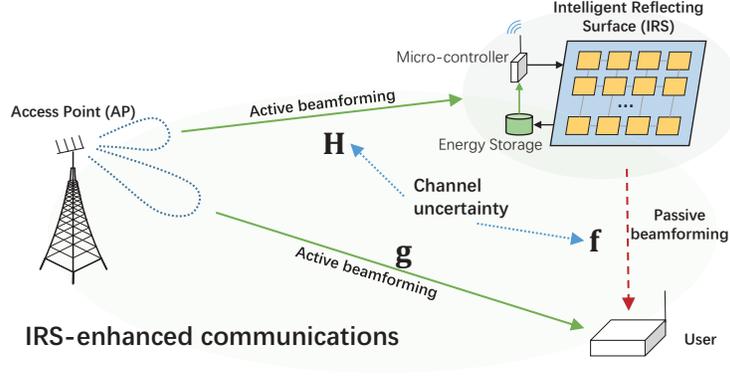}
	\caption{IRS-assisted MISO system.}\label{fig_system_model}
\end{figure}

\subsection{IRS's Power Budget Constraint}
Given the AP's transmit beamforming, the incident signal at the IRS is ${\bf x}={\bf H}^H{\bf w}s$. We assume that each reflecting element is connected to a phase controller and also equipped with an energy harvester that is able to harvest RF energy from the AP's beamforming signals. By tuning the magnitudes of the reflecting coefficients $\bm{\rho}\triangleq \{ \rho_1, \ldots, \rho_1, \ldots, \rho_N \}$, a part $\rho_n^2$ of the incident signal power is reflected to the receiver, while the other part $1-\rho_n^2$ is fed into the energy harvester.
To maintain the IRS's operations, the total harvested energy has to meet the IRS's total power consumption, which leads to the following power budget constraint:
\begin{equation*}
\eta\sum_{n} (1 - \rho_n^2) \lVert{\bf h}^H_n{\bf w}\rVert^2 \geq N\mu,
\end{equation*}
where $\eta$ denotes the energy harvesting efficiency and ${\bf h}_n$ is the channel vector from the AP to the $n$-th reflecting element. Note that the IRS's power consumption relates to the number of reflecting elements and the phase resolution~\cite{Huang2018Large}. Assuming identical bit resolutions for all reflecting elements, the IRS's total power consumption is simply denoted as $N\mu$, where $\mu$ denotes the power consumption of a single reflecting element.

We aim to minimize the AP's transmit power, denoted as $||{\bf w}||^2$, by jointly optimizing the active and passive beamforming strategies, constrained by the IRS's power budget constraint and the receiver's SNR requirement.
\begin{subequations}\label{eq_op}
\begin{align}
\min_{{\bf w}, {\boldsymbol{\Theta}}}~&~||{\bf w}||^2\\
s.t~&~|({\bf g}+ {\bf H}{\bf \Theta}{\bf f})^H{\bf w}|^2 \geq \gamma_1, \label{con_snr} \\
~&~\eta\sum_{n} (1 - \rho_n^2) |{\bf h}^H_n{\bf w}|^2 \geq N\mu,\\
~&~\rho_n\in(0,1) \text{ and } \theta_n \in (0, 2\pi) \quad \forall \, n\in\mathcal{N}.\label{con_power}
\end{align}
\end{subequations}
Note that problem~\eqref{eq_op} can be easily extended to a multi-user case by imposing individual SNR requirement for each user. 

\section{Robust Active and Passive Beamforming Optimization}
It is clear that problem~\eqref{eq_op} is non-convex due to the couplings between ${\bf w}$ and ${\boldsymbol{\Theta}}$, which is conventionally solved in an AO framework. Different from~\cite{18pbf_rui2}, we allow each reflecting element to tune both the magnitude and phase of the reflecting coefficients. Moreover, the optimal solution to~\eqref{eq_op} relies on the knowledge of exact channel information, including the direct channel ${\bf g}$ and the reflected channels $({\bf H}, {\bf f})$ via the IRS, which are inevitably subject to estimation errors. In the sequel, we firstly propose a channel uncertainty model for the IRS-assisted channels, and then reformulate a robust counterpart of the power minimization problem in~\eqref{eq_op}. After that, we transform the robust counterpart to a tractable form that is the major benefit for our algorithm design.

\subsection{Channel Uncertainty Model}

We assume that the direct channel ${\bf g}$ from the HAP to the receiver can be estimated accurately by the active receiver in a training process. In particular, the HAP can send a known pilot information to the receiver with fixed transmit power. Meanwhile, the IRS switches off its reflecting elements. The channel ${\bf g}$ can be recovered at the receiver based on the received signal samples. However, without information decoding capability at the IRS, the channels ${\bf H}$ and ${\bf f}$ have to be estimated at either the AP or the receiver by overhearing the channel response. Similar to the uncertainty modeling for the IRS-receiver channel in~\cite{frameworkconf} and~\cite{robustsecure}, we assume that the AP-IRS channel matrix ${\bf H}$ is subject to estimation errors, i.e.,~${\bf H} = \bar{\bf H} + {\bf \Delta}_{\bf h}$, where $\bar{\bf H}$ denotes the averaged estimate and ${\bf \Delta}_{\bf h}$ denotes the error estimate of the channel matrix ${\bf H}$. The error estimate ${\bf \Delta}_{\bf h}$ has limited power and thus we can define the uncertainty set $\mathbb{U}_{{\bf h}}$ for ${\bf H}$ as follows:
\begin{equation}\label{equ_uncertain_h}
{\bf H} \in \mathbb{U}_{\bf h} \triangleq \{{\bf H}= \bar {\bf H} + {\bf \Delta}_{\bf h}: \textbf{Tr}({\bf \Delta}_{\bf h}^H{\bf \Delta}_{\bf h}) \leq \delta_{\bf h}^2 \},
\end{equation}
where $\textbf{Tr}(\cdot)$ denotes the trace operation and $\delta_{\bf h}$ is the power limit of error estimate ${\bf \Delta}_{\bf h}$ corresponding to the channel ${\bf H}$.

The estimation of the IRS-receiver channel ${\bf f}$ becomes more difficult as the passive IRS cannot emit RF pilot signals for channel training. As such, the estimation of the channel ${\bf f}$ has to be bundled with the AP-IRS channel ${\bf H}$ and performed at the receiver by overhearing a mixture of signals from the AP and the IRS's reflections, e.g.,~\cite{robustframework}. Define the IRS-assisted reflecting channel ${\bf H}_{\bf f}$ as follows:
\[
{\bf H}_{\bf f} \triangleq\text{diag}({\bf f}) {\bf H}= [f_1 {\bf h}_1, f_2 {\bf h}_2,\ldots, f_N{\bf h}_N],
\]
which is the cascaded channel matrix from the AP to the receiver via the IRS~\cite{robustframework}. Hence, the channel model in~\eqref{eq_enhanced} can be rewritten as $\hat {\bf g} = {\bf g} + {\bf H}_{\bf f} {\bf v}$ where ${\bf v} =[\rho_1e^{j\theta_1},\ldots,\rho_N e^{j\theta_N})]^T$ denotes the diagonal vector of the matrix $\bm{\Theta}$. Similar to~\eqref{equ_uncertain_h}, we can define the uncertainty of channel ${\bf H}_{\bf f}$ as follows:
\begin{equation}\label{equ_uncertain_f}
{\bf H}_{\bf f} \in \mathbb{U}_{\bf f} \triangleq \{{\bf H}_{\bf f}= \bar {\bf H}_{\bf f} + {\bf \Delta}_{\bf f}: \textbf{Tr}({\bf \Delta}_{\bf f}^H{\bf \Delta}_{\bf f}) \leq \delta_{\bf f}^2 \},
\end{equation}
where $\delta_{\bf f}$ denotes the power limit of the error estimate ${\bf \Delta}_{\bf f}$ for the reflecting channel ${\bf H}_{\bf f}$. The average channel estimate $\bar {\bf H}_{\bf f}$ and the power limit $\delta_{\bf f}$ are assumed to be known in advance by channel measurements.

\subsection{Robust Counterpart and Reformulations}
Given the channel uncertainty models in~\eqref{equ_uncertain_h} and~\eqref{equ_uncertain_f}, the robust counterpart of~\eqref{eq_op} can be formulated as follows:
\begin{subequations}\label{prob-wc}
\begin{align}
\min_{{\bf w}, {\bf v}}~&~||{\bf w}||^2\\
s.t~&~|({\bf g}+ {\bf H}_{\bf f}{\bf v})^H{\bf w}|^2 \geq \gamma_1,\quad \forall \, {\bf H}_{\bf f}\in\mathbb{U}_{\bf f}, \label{con1-wc}\\
~&~\eta\sum_{n} (1 - \rho_n^2) |{\bf h}^H_n{\bf w}|^2 \geq N\mu, \quad \forall\,  {\bf H}\in\mathbb{U}_{\bf h}\label{con2-wc}\\
~&~\rho_n\in(0,1) \text{ and } \theta_n \in (0, 2\pi) \quad \forall \, n\in\mathcal{N}.\label{con_power_wc}
\end{align}
\end{subequations}
Here the constraints~\eqref{con1-wc} and~\eqref{con2-wc} define the receiver's worst-case SNR requirement and the IRS's worst-case power budget constraint, respectively. To simplify problem~\eqref{prob-wc}, we assume that all the IRS's reflecting elements have the same magnitude ${\rho}$. The simplification allows us to decompose the optimization of the magnitude $\rho$ and the phase vector $\boldsymbol{\theta}\triangleq[e^{j\theta_1},\ldots,e^{j\theta_N}]^T$. In particular, we can rewrite the IRS-enhanced channel as $\hat{\bf g} = {\bf g}+ {\bf H}_{\bf f}{\bf v} = {\bf g}+{  \rho} {\bf H}_{\bf f} {\bm\theta}$. The energy harvested by the IRS is also simplified as $\eta (1 - \rho^2) ||{\bf H}^H{\bf w}||^2$. Hence, we can simplify the robust problem in~\eqref{prob-wc} as follows:
\begin{subequations}\label{prob_robust}
\begin{align}
\min_{{\bf w}, \boldsymbol{\theta},\rho}~&~\lVert {\bf w} \rVert^2\label{robust_obj}\\
s.t.~&~\lvert ({\bf g}+{  \rho} {\bf H}_{\bf f} {\bm\theta})^H{\bf w} \rvert^2 \geq \gamma_1,\quad \forall \, {\bf H}_{\bf f}\in\mathbb{U}_{\bf f},\label{SNR}\\
~&~\eta (1 - \rho^2) \lVert{\bf H}^H{\bf w}\rVert^2 \geq N\mu,\quad \forall \, {\bf H}\in\mathbb{U}_{\bf h},\label{energy}\\
~&~\rho\in(0,1) \text{ and } \theta_n \in (0, 2\pi) \quad \forall \, n\in\mathcal{N}.\label{con_power_simp}
\end{align}
\end{subequations}
The difficulty of problem~\eqref{prob_robust} firstly lies in that the magnitude $\rho$ is coupled with the phase vector ${\bm \theta}$ and the AP's transmit beamforming ${\bf w}$. With perfect channel conditions, the joint optimization of $({\bf w}, \boldsymbol{\theta})$ for any fixed $\rho$ can follow the conventional AO method, similar to that in~\cite{yuze}. Given the fixed ${\bf w}$, the worst-case constraints in~\eqref{SNR}-\eqref{energy} define the lower and upper bounds on $\rho$. This implies that we may resort to a linear search method to optimize $\rho$. Another difficulty comes from the semi-infinite constraints in~\eqref{SNR}-\eqref{energy}, which are required to hold for any channel error estimate in the uncertainty set. In the sequel, we first present a simple heuristic to decompose the optimization of ${\bf w}$ and ${\bm \theta}$. Then, we reformulate the worst-case constraints in~\eqref{SNR}-\eqref{energy} and present a convex approximation to problem~\eqref{prob_robust}.

\subsection{Heuristics for Optimization Decomposition}
In this work, different from the conventional AO method, we propose a new method that decomposes the joint optimization of $({\bf w}, \boldsymbol{\theta})$ and also enable us to rewrite the semi-infinite constraint~\eqref{SNR} into a linear matrix inequality. The intuition behind the decomposition lies in that we can expect that the IRS-assisted channel ${\rho} {\bf H}_{\bf f} {\bm\theta}$ in~\eqref{SNR} aligns with the direct channel ${\bf g}$ from the AP to the receiver. This happens when the direct channel ${\bf g}$ is strong under the line-of-sight (LoS) channel conditions. We present this property as follows:
\begin{observation}\label{prop_SNR}
Assuming a large-scale IRS, i.e.,~$N\gg M$, we can always find a phase vector ${\bf \theta}$ such that
\begin{equation}\label{equ_phase}
  {\bf H}_{\bf f} {\bm \theta}=\kappa {\bf g},
\end{equation}
where $\kappa \in \mathbb{R^+}$ is a scalar constant.
\end{observation}
This property is easy to verify as the phase vector ${\bm \theta}$ of a large-scale IRS provides sufficient control variables to solve a set of linear equations in the form of~${\bf H}_{\bf f} {\bm \theta}=\kappa {\bf g}$ for a feasible $\kappa>0$. Though Proposition~\ref{prop_SNR} may not hold at the optimum of problem~\eqref{prob_robust} due to the coupling of ${\bf w}$ in both~\eqref{SNR} and~\eqref{energy}, it implies that the IRS can always enhance the direct channel ${\bf g}$ by its phase tuning. This may shred some light on {the} optimization decomposition of ${\bf w}$ and ${\bm \theta}$ into two sub-problems. It is clear that the phase solution to~\eqref{equ_phase} is not unique due to the large size of reflecting elements. Based on Proposition~\ref{prop_SNR}, we can simply choose ${\bm \theta}$ to maximize the channel gain $||(1+ {\rho}\kappa_m){\bf g}||^2$ of the IRS-assisted channel. This implies a bisection method to search for the maximum gain, denoted as $\kappa_m$, and the corresponding phase vector ${\bm \theta}$.
Considering the channel uncertainty model in~\eqref{equ_uncertain_f}, practically we can replace ${\bf H}_{\bf f}$ in~\eqref{equ_phase} by its mean estimate $\bar{\bf H}_{\bf f}$.

\subsubsection{Convex Reformulation for Worst-case Constraints}
Given the optimized phase vector ${\bm \theta}$ and the channel gain $\kappa_m$, now we focus on the optimization of $\rho$ and ${\bf w}$ in problem~\eqref{prob_robust}, which are closely coupled with the uncertain channel matrices ${\bf H}_{\bf f}$ and ${\bf H}$. In the sequel, we explore equivalent reformulations of the constraints~\eqref{SNR}-\eqref{energy}, respectively.
\begin{proposition}\label{prop_wc1}
Given the solution $({\bm \theta}, \kappa_m)$ to~\eqref{equ_phase}, the constraint~\eqref{SNR} is equivalent to the following matrix inequality:
\begin{align}\label{equ_cvx_snr}
		\left[\begin{matrix}
			\rho^2\left({\bm \theta}{\bm \theta}^H \otimes{\bf W}\right) + t {\bf I}_{MN} & \alpha\rho({\bm \theta} \otimes {\bf W}){\bf g}\\
			\alpha\rho {\bf g}^H ({\bm \theta}\otimes {\bf W})^H &\alpha^2{\bf g}^H{\bf W}{\bf g}-\gamma_1 - t \delta_{\bf f}^2
		\end{matrix}\right] \succeq 0,
\end{align}
for some $t\geq 0$, where $\alpha\triangleq(1+\rho \kappa_{m})$, ${\bf I}_{MN}$ denotes the identity matrix with size $MN$, and the semidefinite matrix ${\bf W}$ is a rank-one relaxation of ${\bf w}{\bf w}^H$, i.e.,~${\bf W} \succeq {\bf w}{\bf w}^H$.
\end{proposition}
\begin{proof}
With fixed $({\bm \theta}, \kappa_m)$, we can rewrite~\eqref{SNR} as follows:
\begin{equation}\label{robust_c1}
\lvert (1+\rho \kappa_{m}) {\bf g}^H {\bf w}+ \rho \left({\bm \Delta}_{{\bf f}} {\boldsymbol \theta}\right)^H {\bf w} \rvert^2 \geq \gamma_1, \forall \,\,{\bf H}_{\bf f}\in\mathbb{U}_{\bf f}.
\end{equation}
Let $\alpha=(1+\rho \kappa_{m})$ and ${\bf d}_{\bf f} = {\bm \Delta}_{{\bf f}} {\boldsymbol \theta}$ for notational convenience. The LHS of \eqref{robust_c1} can be represented by the following semidefinite relaxation:
\begin{align}
& \lvert (1+\rho \kappa_{m}) {\bf g}^H {\bf w}+ \rho {\bf d}_{\bf f}^H {\bf w} \rvert^2 \nonumber\\
&=\rho^2 {\bf d}_{\bf f}^H {\bf W} {\bf d}_{\bf f} + \alpha^2 {\bf g}^H {\bf W} {\bf g}  + \alpha \rho {\bf g}^H {\bf W} {\bf d}_{\bf f} + \alpha \rho {\bf d}_{\bf f}^H {\bf W} {\bf g} \nonumber \\
&=\rho^2\textbf{Tr}\left( {\bm \Delta}_{\bf f}^H {\bf W} {\bm \Delta}_{\bf f} {\bm \theta}{\bm \theta}^H \right) + \alpha^2 {\bf g}^H {\bf W} {\bf g} \nonumber\\
&\hspace{1cm} + \alpha \rho \textbf{Tr}\left({\bm \Delta}_{\bf f}^H {\bf W} {\bf g}{\bm \theta}^H\right) + \alpha\rho \textbf{Tr}\left({\bm \theta} {\bf g}^H {\bf W} {\bm \Delta}_{\bf f} \right). \label{equ_snr_trace}
\end{align}
Define $\text{vec}(\cdot)$ as the vectorization operation of a matrix by cascading all column vectors into one single vector. Thus, we have ${\bf Tr}({\bf A}^H{\bf B})=\text{vec}({\bf A})^H\text{vec}({\bf B})$ and $\text{vec}({\bf A}{\bf B}{\bf C})=({\bf C}^H\otimes{\bf A}) \text{vec}({\bf B})$, where $\otimes$ denotes the kronecker product of two matrices with proper sizes. As such, we can further rewrite~\eqref{equ_snr_trace} as follows:
\begin{equation*}
\begin{aligned}
&\lvert (1+\rho \kappa_{m}) {\bf g}^H {\bf w}+ \rho {\bf d}_{\bf f}^H {\bf w} \rvert^2\\
&=\rho^2 \text{vec}({\bm \Delta}_{\bf f})^H ( {\bm \theta}{\bm \theta}^H \otimes {\bf W} ) \text{vec}({\bm \Delta}_{\bf f}) + \alpha^2 {\bf g}^H {\bf W} {\bf g}\\
&\quad \, + \alpha \rho {\bf g}^H( {\bm \theta} \otimes {\bf W} )^H \text{vec}({\bm \Delta}_{\bf f}) + \alpha \rho \text{vec}({\bm \Delta}_{\bf f})^H ( {\bm \theta} \otimes {\bf W} ){\bf g}.
\end{aligned}	
\end{equation*}
Till now, the LHS of~\eqref{robust_c1} can be rewritten in a quadratic form, i.e.,~${\bf z}^H {\bf M}{\bf z} \geq 0$, where ${\bf z} = \left[\begin{matrix}
\text{vec}({\bm \Delta}_{\bf f}) \\
1
\end{matrix}\right]$ and the matrix coefficient ${\bf M}$ is given as follows:
\begin{equation*}
{\bf M}=\left[\begin{matrix}
		\rho^2\left({\bm \theta}{\bm \theta}^H \otimes{\bf W}\right) & \alpha\rho ({\bm \theta} \otimes {\bf W}){\bf g}\\
		\alpha\rho {\bf g}^H ({\bm \theta}\otimes {\bf W})^H &\alpha^2{\bf g}^H{\bf W}{\bf g}-\gamma_1
	\end{matrix}\right] .
\end{equation*}
It is easy to see from~\eqref{robust_c1} that ${\bf z}^H {\bf M}{\bf z} \geq 0$ holds for any ${\bm \Delta}_{\bf f}$ satisfying $\text{vec}({\bm \Delta}_{\bf f})^H\text{vec}({\bm \Delta}_{\bf f}) \leq \delta_{\bf f}^2$. By S-Lemma~\cite{quadratic}, we can always find some $t\geq 0$ such that
${\bf M}\succeq  t \left[\begin{matrix}
		-{\bf I}_{MN} & {\bf 0}\\
		{\bf 0} &  \delta_{\bf f}^2
	\end{matrix}\right]$,
which is equivalent to the matrix inequality in~\eqref{equ_cvx_snr}.
\end{proof}
Proposition~\ref{prop_wc1} transforms the worst-case constraint in~\eqref{SNR} into a semidefinite matrix inequality. Note that the matrix coefficient ${\bm \theta}{\bm \theta}^H \otimes{\bf W}$ can be further simplified as $({\bm \theta}\otimes{\bf W})( {\bm \theta}^H \otimes {\bf I}_M)$. The common term ${\bm \theta}\otimes{\bf W}$ in~\eqref{equ_cvx_snr} is linear with respect to
the beamforming matrix ${\bf W}$. However, the resulting constraint in~\eqref{equ_cvx_snr} is still non-convex due to the quadratic coupling between $\rho$ and ${\bf W}$. We also observe that for any fixed $\rho$ the constraint in~\eqref{equ_cvx_snr} becomes a linear matrix inequality, which is convex in terms of ${\bf W}$ and the auxiliary variable $t$.
\begin{proposition}\label{prop_energy}
Given the solution $({\bm \theta}, \kappa_m)$ to~\eqref{equ_phase}, the constraint in~\eqref{energy} has the following semidefinite reformulation:
\begin{align}\label{equ_cvx_energy}
\left[\begin{array}{cc}
	{\bf W}_c+ {\tau}{\bf I}_{MN},&  {\bf W}_c \text{vec}(\bar{\bf H})\\
	\text{vec}(\bar{\bf H})^H{\bf W}_c, &  \bar \gamma_0 -\frac{N\mu}{\eta(1-\rho^2)}-\tau {\delta_{\bf h}^2}
	\end{array}\right]	{  \succeq 0},
\end{align}
for some $\tau \geq 0$, where we define ${\bf W}_c = {\bf I}_N \otimes{\bf W}$ and $ \bar \gamma_0 =  \text{vec}(\bar{\bf H})^H{\bf W}_c \text{vec}(\bar{\bf H})$ for notational convenience.
\end{proposition}
\begin{proof}
The proof of Proposition~\ref{prop_energy} follows a similar idea {to} that for Proposition~\ref{prop_wc1} by rewriting the semi-infinite constraint~\eqref{energy} {in} a quadratic form. Let $\gamma_0 = \lVert {\bf H}^H {\bf w} \rVert^2$ and $\bar \gamma_0 = \lVert \bar {\bf H}^H {\bf w} \rVert^2$. We have the following reformulations:
\begin{equation*}
\begin{aligned}
\gamma_0 &= {\bf w}^H \left( {\bf \bar{H}} + {\bm \Delta}_{\bf h} \right) \left( {\bf \bar{H}} + {\bm \Delta}_{\bf h} \right)^H {\bf w} \\
&=  {\bf w}^H {\bm \Delta}_{\bf h} {\bm \Delta}_{\bf h}^H {\bf w} + {\bf w}^H {\bm \Delta}_{\bf h} {\bf \bar{H}}^H {\bf w} + {\bf w}^H {\bf \bar{H}}{\bm \Delta}_{\bf h}^H {\bf w} + \bar \gamma_0\\
&= {{\bf Tr}({\bm \Delta}_{\bf h}^H{\bf W}{\bm \Delta}_{\bf h})} + {\bf Tr}({\bm \Delta}_h^H{\bf W}{\bf \bar{H}} ) + {\bf Tr}({\bf \bar{H}}^H{\bf W}{\bm \Delta}_{\bf h})+  \bar \gamma_0.
\end{aligned}
\end{equation*}
Similarly, by using the trace equalities, i.e.,~${\bf Tr}({\bf A}^H{\bf B})=\text{vec}({\bf A})^H\text{vec}({\bf B})$ and $\text{vec}({\bf A}{\bf B})=\left({\bf I}_N\otimes{\bf A}\right){\bf B}$ for the matrices ${\bf A}$ and ${\bf B}$ with the dimensions $M\times M$ and $M\times N$, respectively, we can rewrite $\gamma_0 $ as follows:
\begin{align}
\gamma_0& = \text{vec}({\bm \Delta}_{\bf h})^H{\bf W}_c\text{vec}( {\bm \Delta}_{\bf h}) + \text{vec}({\bm \Delta}_{\bf h})^H{\bf W}_c{\text{vec}(\bar{\bf H})}\nonumber\\
&\hspace{2.5cm} + \text{vec}({\bf \bar{H}})^H{\bf W}_c\text{vec}({\bm \Delta}_{\bf h})+  \bar \gamma_0,
\end{align}
where ${\bf W}_c = {\bf I}_N \otimes{\bf W}$. Note that we can also write $\bar \gamma_0$ as $\bar \gamma_0 = \text{vec}(\bar{H})^H{\bf W}_c \text{vec}(\bar{\bf H})$. As such, we can reformulate the constraint \eqref{energy} into a quadratic form of the uncertain vector $\text{vec}({\bm \Delta}_{\bf h})$. Note that $\text{vec}({\bm \Delta}_{\bf h})^H \text{vec}( {\bm \Delta}_{\bf h}) = {\bf Tr}({\bf \Delta}_{\bf h}^H{\bf \Delta}_{\bf h}) \leq \delta_{\bf h}^2$. Similarly by S-Lemma, the worst-case power budget constraint~\eqref{energy} can be equivalent to~\eqref{equ_cvx_energy}.
\end{proof}

\subsubsection{Iterative Search Algorithm for $(\rho, {\bf W})$}
Till now, we can reformulate the robust power minimization problem under {the} uncertain channel conditions as follows:
\begin{equation}\label{prob_robust_1}
\min_{\rho\geq0,{\bf W}\succeq {\bf 0}, t\geq 0,\tau\geq 0}\{ {\bf Tr}({\bf W}) :\, \eqref{equ_cvx_snr} \text{ and } \eqref{equ_cvx_energy}\}
\end{equation}
The linear beamforming vector~${\bf w}$ can be retrieved by eigenvalue decomposition if the matrix solution ${\bf W}$ to problem~\eqref{prob_robust_1} is of rank one. Otherwise we can extract an approximate rank-one solution via Gaussian randomization~\cite{SDR}. However, problem~\eqref{prob_robust_1} is still non-convex due to the coupling of $\rho$ and ${\bf W}$ in the constraints \eqref{equ_cvx_snr} and \eqref{equ_cvx_energy}. For fixed $\rho$, it is easy to verify that both constraints~\eqref{equ_cvx_snr} and~\eqref{equ_cvx_energy} become linear matrix inequalities. Thus, problem~\eqref{prob_robust_1} can be efficiently solved by semidefinite programming. However, with fixed ${\bf W}$, it is still difficult to optimize $\rho$ due to to the non-convex structure in constraint~\eqref{equ_cvx_snr}. This implies that the conventional AO method does not apply to problem~\eqref{prob_robust_1} directly. In the sequel, we first exploit the structural property of problem~\eqref{prob_robust} and then devise a simple iterative algorithm to search for $(\rho, {\bf W})$.


\begin{proposition}\label{prop_structure}
Assuming that problem~\eqref{prob_robust} is feasible, the constraint in~\eqref{energy} always holds with equality at optimum of problem~\eqref{prob_robust}.
\end{proposition}
\begin{proof}
Considering a non-trivial case when the feasible region of $\rho$ is non-empty, it is clear that the constraints in~\eqref{SNR} and~\eqref{energy} define the lower and upper bounds on $\rho$, for any fixed ${\bf w}$, {respectively}. It is easy to verify that this relationship also holds in~\eqref{equ_cvx_snr} and~\eqref{equ_cvx_energy}. In the following, we prove this proposition by contradiction. In particular, we assume that $(\rho, {\bf w})$ is optimal to problem~\eqref{prob_robust} and the constraint~\eqref{energy} holds with strict inequality. Then, we have different strategies to improve the objective as follows:
\begin{enumerate}
  \item If~\eqref{SNR} holds with strict inequality, it is clear that we can properly scale ${\bf w}$ by a scalar factor $s<1$ such that both constraints in~\eqref{SNR} and~\eqref{energy} still hold. The scaling $s{\bf w}$ apparently leads to a better objective in~\eqref{robust_obj}, which becomes contradictory to our assumption.
  \item If~\eqref{SNR} holds with equality, we can still construct a new solution leading to a reduced transmit power. In this case, the inequality~\eqref{SNR} always holds when we increase $\rho$. Let $\rho_{\min}$ and $\rho_{\max}$ denote the lower and upper bounds of $\rho$ with the fixed ${\bf w}$, {respectively}. Hence, we have $\lvert ({\bf g}+{  \rho_{\min}} {\bf H}_{\bf f} {\bm\theta})^H{\bf w} \rvert^2 = \gamma_1$ and $\eta (1 - \rho^2_{\max}) \lVert{\bf H}^H{\bf w}\rVert^2 =  N\mu$ in the worst-case channel conditions. We can simply set $\rho_{m} = (\rho_{\min}+\rho_{\max})/2$ to ensure strict inequalities in~\eqref{SNR}-\eqref{energy}. This becomes exactly the first case, in which we can scaling down ${\bf w}$ and achieve a reduced transmit power.
\end{enumerate}
It is clear that in any of the above two cases we can improve the objective function, which implies that the constraint in~\eqref{energy} holds with equality at optimum.
\end{proof}

Proposition~\ref{prop_structure} implies a simple iterative solution method for problem~\eqref{prob_robust_1}. Specifically, the algorithm starts with a feasible $\rho$, the optimal transmit beamforming ${\bf W}$ can be efficiently optimized by solving the semidefinite program in~\eqref{prob_robust_1}. According to Proposition~\ref{prop_structure}, we can simply update $\rho$ by its maximum defined by the constraint in~\eqref{equ_cvx_energy}. The detailed procedures are listed in Algorithm~\ref{alg_iterative}. The algorithm terminates when the AP's transmit power becomes stabilized.

\begin{algorithm}
	\caption{The Max-$\rho$ Algorithm for Problem~\eqref{prob_robust_1}}\label{alg_iterative}
	\begin{algorithmic}[1]
		\STATE Initial with small $\rho$, $k\leftarrow 1$, $\epsilon\leftarrow 10^{-5}$
		\STATE \textbf{while} $\lvert {\bf Tr}({\bf W}^{(k)})   -  {\bf Tr}( {\bf W}^{(k-1)}) \rvert \geq \epsilon$
		\STATE \hspace{5mm} $k\leftarrow k+1$
		\STATE \hspace{5mm} Find $\kappa_m$ and ${\bm \theta}$ via bisection
		\STATE \hspace{5mm} Solve \eqref{prob_robust_1} with given $\rho$
		\STATE \hspace{5mm} Retrieve ${\bf W}$ and update ${\bf W}^{(k)}\leftarrow {\bf W}$
        \STATE \hspace{5mm} Evaluate the upper bound $\rho_{\max}$
		\STATE \hspace{5mm} Update $\rho\leftarrow \rho_{\max}$
		\STATE \textbf{end while}
	\end{algorithmic}
\end{algorithm}

\section{Numerical Results}

\begin{figure}[b]
	\centering
	\includegraphics[width=\singlesize\textwidth]{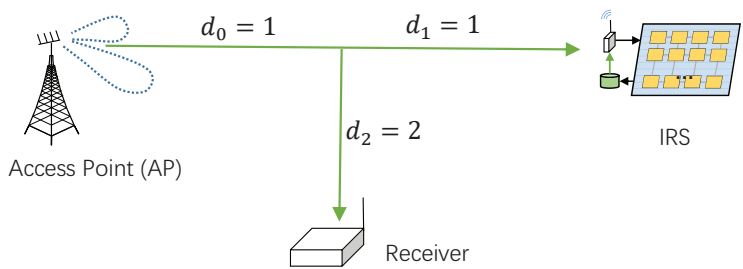}
	\caption{Simulation topology}\label{fig-topo}
\end{figure}
In the simulation, we evaluate the AP's transmit power with different SNR requirements at the receiver. The impact of channel uncertainty on the AP's minimum transmit power is also examined under different parameter settings. Specifically, we consider the AP with $2-4$ antennas and the IRS with $20-100$ reflecting elements. We consider a fixed topology shown in Fig.~\ref{fig-topo} to verify the proposed algorithm. The path loss follows a log-distance propagation model with the loss exponent equal to 2. The path loss at 1 meter distance is 30~dB. To characterize the level of channel uncertainty, we define \emph{uncertainty factors} as $\beta_{\bf h} \triangleq \delta_{\bf h}^2 / \textbf{Tr}(\bar{\bf H}\bar{\bf H}^H)$ and $\beta_{\bf f} \triangleq \delta_{\bf f}^2 / \textbf{Tr}(\bar{\bf H}_{\bf f}\bar{\bf H}_{\bf f}^H)$, respectively for the uncertain channels ${\bf H}$ and ${\bf H}_{\bf f}$. For simplicity, we consider $\beta = \beta_{\bf h} = \beta_{\bf f}$ in the simulation. A larger $\beta$ implies a higher variation of the channel conditions and thus larger errors in channel estimation.

In the following, we firstly verify the convergence of the proposed Max-$\rho$ Algorithm and explain its efficacy. Secondly, we conduct a set of experiments to study the impact of different parameters on the AP's minimum transmit power, including a) the uncertainty factor, b) the size of the IRS, c) the SNR requirement at the receiver, and d) the number of AP's antennas. For each simulation setting, we run the experiment 10 times with randomly generated channel conditions and record the averaged performance for a fair comparison. In Fig.~\ref{fig-converge}, we show the convergence of the AP's transmit power and the IRS's magnitude $\rho$ of reflecting coefficients in Algorithm~\ref{alg_iterative}. The magnitude $\rho$ is also called a power-splitting (PS) ratio. We consider $M=2$ antennas at the AP and $N=20$ reflecting elements in the IRS. The SNR requirement at the receiver is set to $\gamma_1=30$~dB and the uncertainty factor is set to $\beta=0.1$. It is clear that initially the AP operates with a large transmit power to ensure the fulfillment of the worst-case SNR requirement at the receiver and the worst-case power budget constraint at the IRS. As the algorithm iterates, the AP's transmit power decreases significantly meanwhile the PS ratio $\rho$ increases to reflect more RF power to the receiver. By dynamically adjusting the operating parameters at both the AP and the IRS in an alternating manner, the AP can tune down its transmit power gradually while still maintaining the desired service provisioning requirement. We can see from Fig.~\ref{fig-converge} that Algorithm~\ref{alg_iterative} converges quickly in around 20 iterations, which verifies the effectiveness of our algorithm design.

\begin{figure}[t]
	\centering
	\includegraphics[width=\singlesize\textwidth]{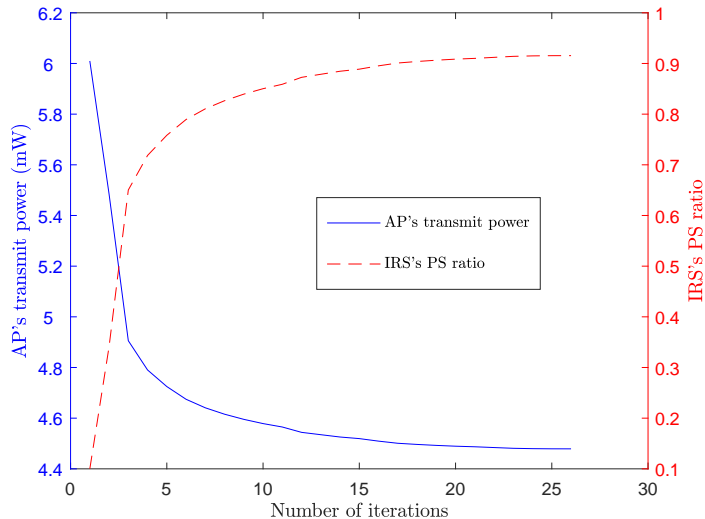}
	\caption{Convergence of Algorithm~\ref{alg_iterative} in a few iterations.}\label{fig-converge}
\end{figure}

\begin{figure}[t]
	\centering
	\includegraphics[width=\singlesize\textwidth]{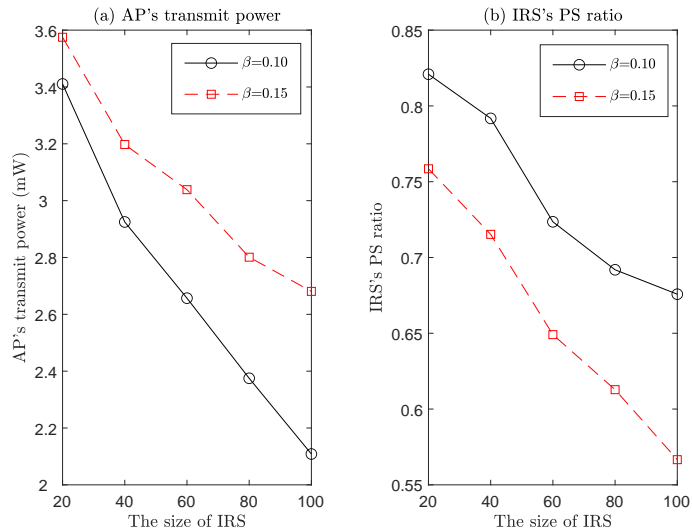}
	\caption{Performance impact of uncertainty factor.}\label{fig-uf}
\end{figure}

\begin{figure}[t]
	\centering
	\includegraphics[width=\singlesize\textwidth]{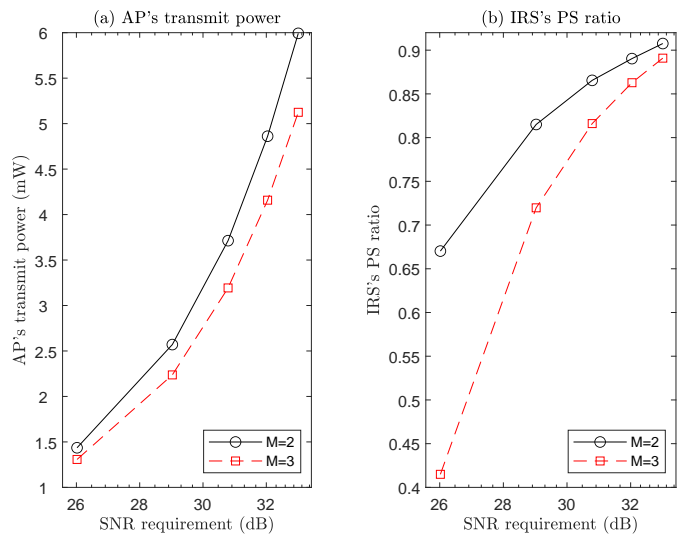}
	\caption{Performance impact of user's SNR requirement.}\label{fig-snr}
\end{figure}

\begin{figure}[t]
	\centering
	\includegraphics[width=\singlesize\textwidth]{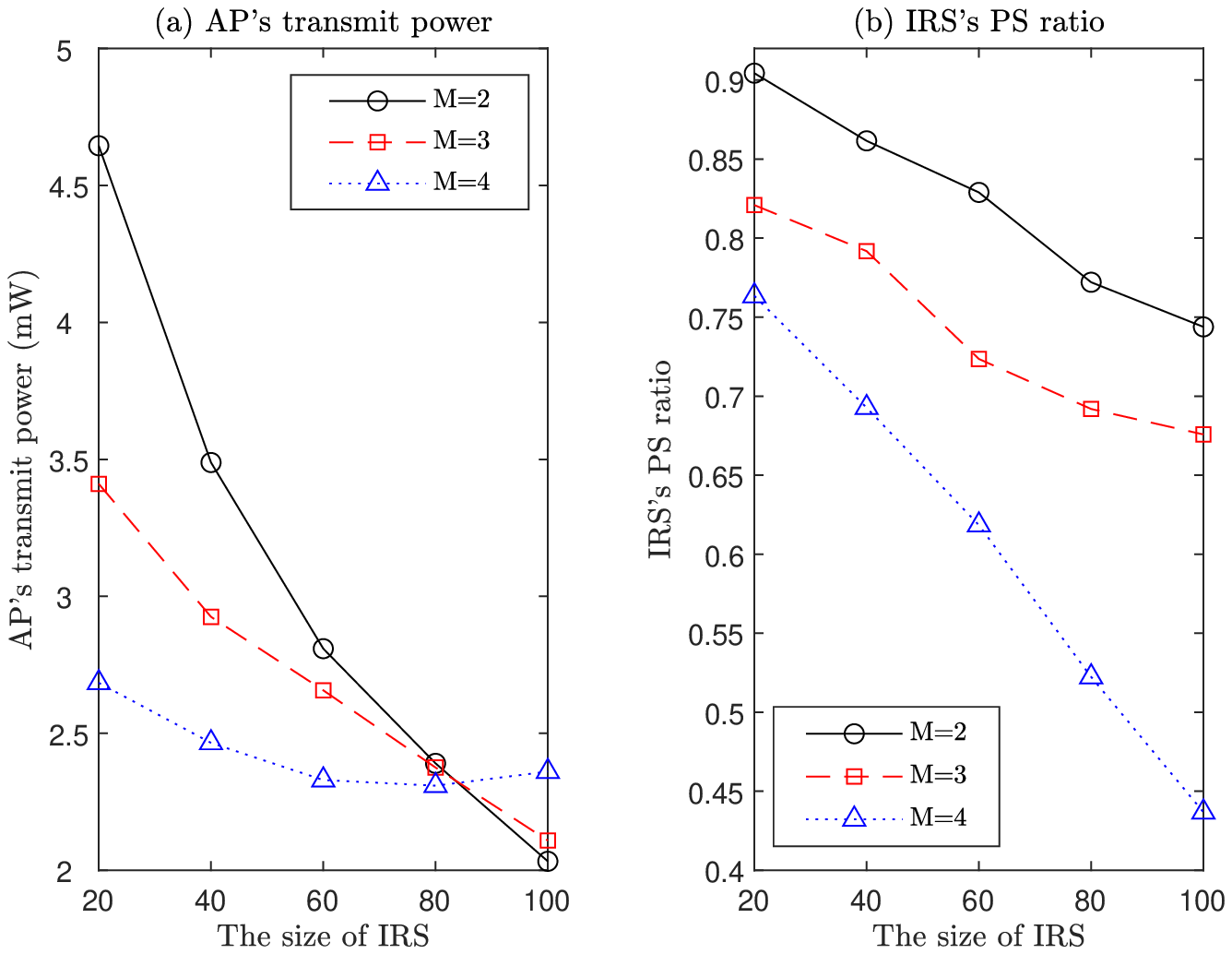}
	\caption{Performance impact of the number of AP's antennas.}\label{fig-antenna}
\end{figure}

In Fig.~\ref{fig-uf}, we evaluate the impact of the uncertainty factor $\beta$ on the AP's transmit power. The uncertainty factors are set to $\beta=0.1$ and $\beta=0.15$, respectively. The number of the IRS's reflecting elements varies from 20 to 100. As shown in Fig.~\ref{fig-uf}(a), with a fixed size of the IRS, the AP's transmit power becomes higher with a larger uncertainty level. This observation is intuitive since the AP needs to raise its transmit power to ensure the worse worst-case data rate guarantee at the receiver. This can be viewed as the price for robustness. Besides, as shown in Fig.~\ref{fig-uf}(b), when the uncertainty factor $\beta$ becomes large, the IRS prefers to reserve more energy by setting a smaller value of $\rho$ to maintain its operations even in the worst-case channel conditions. Fig.~\ref{fig-uf} also verifies that a large-size IRS can provide more performance improvement. In particular, the AP's transmit power decreases as the number of reflecting elements increases. A larger size of IRS indicates more channel diversity that can be exploited by the AP to enhance the user's data rate. This means that we can maintain the same service provisioning with reduced transmit power at the AP. On the other hand, a larger size of IRS also implies more energy consumption, which requires more energy harvesting from the AP's signal beamforming by setting a smaller value of $\rho$, as shown in Fig.~\ref{fig-uf}(b).

We further exam the AP's transmit power by varying the receiver's SNR requirement. We fix the number of reflecting elements at $N=50$ and the uncertainty factor at $\beta=0.1$. As shown in Fig.~\ref{fig-snr}(a), the AP's minimum transmit power grows with the increasing of the user's SNR requirement. However, the PS ratio $\rho$ increases much slower as shown in Fig.~\ref{fig-snr}(b). The IRS's PS ratio can be jointly tuned up to enhance the information transmission when the user's SNR requirement becomes more stringent. In Fig.~\ref{fig-antenna}, we investigate the performance impact of the number of AP's antennas and the size of the IRS. The uncertainty factor is fixed at $\beta=0.1$. As shown in the Fig.~\ref{fig-antenna}(a), the AP's transmit power decreases with the increasing in the size of IRS, which is consistent with the results in Fig.~\ref{fig-uf}. We also record an interesting observation in Fig.~\ref{fig-antenna}(a) where different curves intersect with each. This implies that a larger size of the IRS and more antennas at the AP do not guarantee a better performance under channel uncertainties. The reason is that a large-size IRS and more antennas at the AP may exaggerate the uncertainty in the cascaded AP-IRS-receiver channel, and thus the price of robustness also becomes higher. Moreover, the IRS's power budget constraint also limits the feasible size of IRS in practice. These observations provide useful insights for practical deployment of IRS with imperfect channel and energy conditions.

\section{Conclusions}
In this paper, we have proposed a robust power minimization problem by jointly optimizing the active and passive beamforming strategies, subject to both the worst-case data rate requirement at the receiver and the worst-case power budget constraint at the IRS. We have reformulated the worst-case
constraints into matrix inequalities and devised a novel iterative algorithm to search for the AP's transmit beamforming and the IRS's power-splitting ratio. Our simulation results reveal that the price of robustness is inevitable, however it can be alleviated by using a larger-size IRS. In our future work, we may focus on a quantitative study on the relationship between the size of IRS and the price of robustness.

\bibliographystyle{IEEEtran}

\bibliography{robust-irs}

\begin{thebibliography}{10}
\providecommand{\url}[1]{#1}
\csname url@samestyle\endcsname
\providecommand{\newblock}{\relax}
\providecommand{\bibinfo}[2]{#2}
\providecommand{\BIBentrySTDinterwordspacing}{\spaceskip=0pt\relax}
\providecommand{\BIBentryALTinterwordstretchfactor}{4}
\providecommand{\BIBentryALTinterwordspacing}{\spaceskip=\fontdimen2\font plus
\BIBentryALTinterwordstretchfactor\fontdimen3\font minus
  \fontdimen4\font\relax}
\providecommand{\BIBforeignlanguage}[2]{{%
\expandafter\ifx\csname l@#1\endcsname\relax
\typeout{** WARNING: IEEEtran.bst: No hyphenation pattern has been}%
\typeout{** loaded for the language `#1'. Using the pattern for}%
\typeout{** the default language instead.}%
\else
\language=\csname l@#1\endcsname
\fi
#2}}
\providecommand{\BIBdecl}{\relax}
\BIBdecl

\bibitem{18scm_ian1}
C.~{Liaskos}, S.~{Nie}, A.~{Tsioliaridou}, A.~{Pitsillides}, S.~{Ioannidis},
  and I.~{Akyildiz}, ``A new wireless communication paradigm through
  software-controlled metasurfaces,'' \emph{IEEE Commun. Mag.}, vol.~56, no.~9,
  pp. 162--169, Sep. 2018.

\bibitem{19survey_renzo}
M.~D. Renzo, M.~Debbah, D.~T.~P. Huy, A.~Zappone, M.~Alouini, C.~Yuen,
  V.~Sciancalepore, G.~C. Alexandropoulos, J.~Hoydis, H.~Gacanin, J.~de~Rosny,
  A.~Bounceu, G.~Lerosey, and M.~Fink, ``Smart radio environments empowered by
  {AI} reconfigurable meta-surfaces: An idea whose time has come,''
  \emph{EURASIP J. Wireless Commun. Network.}, vol. 129, 2019.

\bibitem{joint_overview}
Q.~Wu and R.~Zhang, ``Towards smart and reconfigurable environment: Intelligent
  reflecting surface aided wireless network,'' \emph{IEEE Commun. Mag.},
  vol.~58, no.~1, pp. 106--112, Jan. 2020.

\bibitem{18scm_ian3}
L.~Christos, N.~Shuai, T.~Ageliki, P.~Andreas, I.~Sotiris, and I.~F. Akyildiz,
  ``Realizing wireless communication through software-defined hypersurface
  environments,'' in \emph{proc. IEEE Int. Sym. WoWMoM}, Jun. 2018.

\bibitem{powermax-swipt}
Q.~Wu and R.~Zhang, ``Weighted sum power maximization for intelligent
  reflecting surface aided {SWIPT},'' \emph{IEEE Wireless Commun. Lett.}, pp.
  1--1, Dec. 2019.

\bibitem{Huang2018Large}
C.~Huang, A.~Zappone, G.~C. Alexandropoulos, M.~Debbah, and C.~Yuen,
  ``Reconfigurable intelligent surfaces for energy efficiency in wireless
  communication,'' \emph{IEEE Trans. Wireless Commun.}, vol.~18, no.~8, Jun.
  2019.

\bibitem{Wu2018Intelligent}
Q.~Wu and R.~Zhang, ``Intelligent reflecting surface enhanced wireless network:
  Joint active and passive beamforming design,'' in \emph{Proc. IEEE Globecom},
  Abu Dhabi, United Arab Emirates, Dec. 2018.

\bibitem{18pbf_rui2}
------, ``Intelligent reflecting surface enhanced wireless network via joint
  active and passive beamforming,'' \emph{IEEE Trans. Wireless Commun.},
  vol.~18, no.~11, pp. 5394--5409, Nov. 2019.

\bibitem{Yu2019Enabling}
X.~Yu, D.~Xu, and R.~Schober, ``Enabling secure wireless communications via
  intelligent reflecting surfaces,'' in \emph{proc. IEEE GLOBECOM}, Dec. 2019.

\bibitem{wang2019eefficient}
Q.~Wang, F.~Zhou, and R.~Q. Hu, ``Energy-efficient beamforming and cooperative
  jamming in {IRS}-assisted {MISO} networks,'' in \emph{proc. IEEE ICC}, Jun.
  2019.

\bibitem{ss_asym}
M.~{Jung}, W.~{Saad}, and G.~{Kong}, ``Spectral efficiency in large intelligent
  surfaces: Asymptotic analysis under pilot contamination,'' in \emph{proc.
  IEEE GLOBECOM}, Dec. 2019, pp. 1--6.

\bibitem{frameworkconf}
G.~Zhou, C.~Pan, H.~Ren, K.~Wang, M.~D. Renzo, and A.~Nallanathan, ``Robust
  beamforming design for intelligent reflecting surface aided {MISO}
  communication systems,'' \emph{arXiv preprint arXiv:1911.06237}, 2019.

\bibitem{robustsecure}
X.~Yu, D.~Xu, Y.~Sun, D.~W.~K. Ng, and R.~Schober, ``Robust and secure wireless
  communications via intelligent reflecting surfaces,'' \emph{arXiv preprint
  arXiv:1912.01497}, 2019.

\bibitem{robustframework}
G.~Zhou, C.~Pan, H.~Ren, K.~Wang, and A.~Nallanathan, ``A framework of robust
  transmission design for {IRS}-aided {MISO} communications with imperfect
  cascaded channels,'' \emph{arXiv preprint arXiv:2001.07054}, 2020.

\bibitem{yuze}
Y.~Zou, S.~Gong, J.~Xu, W.~Cheng, D.~T. Hoang, and D.~Niyato, ``Joint energy
  beamforming and optimization for intelligent reflecting surface enhanced
  communications,'' in \emph{proc. IEEE WCNC Workshops}, Apr. 2020, pp. 1--6.

\bibitem{quadratic}
Z.-Q. Luo, J.~F. Sturm, and S.~Zhang, ``Multivariate nonnegative quadratic
  mappings,'' \emph{SIAM J. on Optimiz.}, vol.~14, no.~4, pp. 1140--1162, 2004.

\bibitem{SDR}
Z.-Q. Luo, W.~{Ma}, A.~M. {So}, Y.~{Ye}, and S.~{Zhang}, ``Semidefinite
  relaxation of quadratic optimization problems,'' \emph{IEEE Signal Process.
  Mag.}, vol.~27, no.~3, pp. 20--34, 2010.

\end{thebibliography}

\end{document}